\documentclass[10pt,conference]{IEEEtran}
\usepackage{amsthm,amssymb,amsmath}
\usepackage{graphicx}
\usepackage{subfigure}
\usepackage{authblk}
\usepackage{balance}  
\usepackage[linesnumbered,ruled,vlined]{algorithm2e}
\usepackage[colorlinks, citecolor=blue]{hyperref}

\newtheorem{theorem}{Theorem}
\newtheorem{definition}{Definition}

\begin{document}
\title{Integrity Verification for Outsourcing  Uncertain Frequent Itemset Mining}
\author[]{Qiwei Lu}
\author[]{Wenchao Huang}
\author[]{Yan Xiong}
\author[]{Xudong Gong}

\affil[]{ University of Science and Technology of China \authorcr \{luqiwei, lzgxd\}@mail.ustc.edu.cn \authorcr \{yxiong, huangwc\}@ustc.edu.cn}

\maketitle

\begin{abstract}
In recent years, due to the wide applications of uncertain data (e.g., noisy data), uncertain frequent itemsets (UFI) mining over uncertain databases has attracted much attention, which differs from the corresponding deterministic problem from the  generalized definition and resolutions.
As the most costly task in association rule mining process, it has been shown that outsourcing this task to a service provider (e.g.,the third cloud party) brings several benefits to the data owner such as cost relief and a less commitment to storage and computational resources. However, the correctness integrity of mining results can be corrupted if the service provider is with random fault or not honest (e.g., lazy, malicious, etc). Therefore, in this paper, we focus on the integrity and verification issue in UFI mining problem during outsourcing process, i.e., how the data owner verifies the mining results. Specifically, we explore and extend the existing  work on deterministic FI outsourcing verification to uncertain scenario. For this purpose, We extend the existing outsourcing FI mining work to uncertain area w.r.t. the two popular UFI definition criteria and the approximate UFI mining methods. Specifically, We construct and improve the  basic/enhanced verification scheme with such different UFI definition respectively. After that, we further discuss the scenario of existing approximation UFP mining, where we can see that our technique can provide good probabilistic guarantees about the correctness of the verification. Finally, we present the comparisons and analysis on the schemes proposed in this paper. 

\end{abstract}

\section{Introduction}
Association rule mining discovers correlated itemsets that occur frequently in a database and is one of the most popular data mining methods.  The problem can be divided into subproblems in two phases \cite{agrawal1994fast}: (i) computing the set of frequent itemsets(FI), (ii) computing the set of association rules based on the mined frequent itemsets.
We address the significance of FI mining phase from the two aspects below:
\begin{itemize}
  \item \textbf{efficiency.} The  FI mining phase has an exponential time complexity and costly \cite{agrawal1994fast}, while  the latter, rule generation, is computational inexpensive.
  \item \textbf{security and integrity.} The integrity and correctness of the FI mining result shapes the base of the latter  associate rules utility and is the crux of  result utility.
\end{itemize}

The efficiency consideration motivates businesses to outsource the task of FI mining to cloud service providers, who undertake the computation process and finally returns the set of frequent itemsets together with their support counts. It has been shown that such outsourcing computation in cloud computing model can brings a number of benefits such as database update, multiple source data mining scenario.
In order to make  the second aspect, i.e., security and integrity,  practical and satisfactory. The existing research works  contribute to practical  encryption/mapping scheme on data content and mining results security \cite{wong2007security}. Regarding the integrity problem, \cite{Fischer-2006-PhdThesis,wong2009audit} proposed the resolutions to ensure the correctness and completeness of the FI mining results.
We note that the work on security and integrity verification outsourcing FI mining is rare, and challenging.

However, all the works above all contribute on the FI outsourcing verification in the deterministic scenario. In other words, they didn't consider the intrinsic noise and uncertainness of the data to be mined, which should be considered seriously with the increasingly popularity of the uncertain data mining need. Such uncertain data mining and computation research originate and contribute to the practical noise-tolerant scenarios including sensor network monitoring \cite{li2009underground,liu2010passive},  moving object search \cite{chen2005robust,cheng2004querying} and so on.

Unlike the deterministic case,  mining the  uncertain frequent itemset (UFI) is more difficult because the support count has to rely on the existence possibility of the items. In fact, there exist two different semantic explanations on UFI, that is, expected support-based frequent itemset \cite{chui2007mining} and probabilistic frequent itemset \cite{bernecker2009probabilistic}.  It has been shown that, as a generalized form and extension of deterministic FI mining,  UFI mining is more complicated and can be divided into: exact expected support-based UFI mining \cite{aggarwal2009frequent,chui2008decremental,chui2007mining,leung2008tree}, exact possible world semantics (PWS) based UFI mining \cite{bernecker2009probabilistic,sun2010mining},  and approximate Poisson, Normal distribution based UFI mining \cite{calders2010approximation,wang2010accelerating} resolutions.
To our best knowledge, the research on uncertain frequent itemsets(UFI) mining in outsouring environment is still a blank.

The first step towards solving the integrity problem is to understand the behavior of a potentially abnormal service provider that can undermine the integrity of the mining results. A cloud service provider may return inaccurate results if (i) it is honest but sloppy, e.g., there are bugs in its mining programs or suffer from random faults; (ii) it is lazy and tries to reduce costly computation, e.g., it mines only a small portion of the dataset; (iii) it is malicious and purposely returns wrong results, e.g., a business competitor has paid the service provider for providing incorrect results so as to affect the business decisions of the data owner. The concept of result integrity should thus be defined on two criteria:
\begin{itemize}
  \item Correctness: All returned frequent itemsets are actually frequent and their returned support counts are correct.
  \item Completeness: All actual frequent itemsets are included in the result.
\end{itemize}

It is worth mentioning that the completeness aspect verification is based on the correctness aspect, with the help of the maximal FI techniques we can complete the completeness verification similar to \cite{Fischer-2006-PhdThesis}, thus in this paper, our focus is the verification of the correctness. In this paper, we aim to define the soundness of our verification schemes against the two different abnormal levels as follows:
\begin{itemize}
  \item \emph{Random fault/ stupid cloud attack tolerant level}: which verifies the abnormal results caused by the random faults  or the attack raised by stupid cloud service provider who is unaware of the verification mechanism;
  \item \emph{Smart cloud attack tolerant level}: which can verifies the abnormal result returned by the attack raised by smart cloud  service provider who is knowledge of the verification mechanism well.
\end{itemize}
Now we summarize the contribution in this paper as follows:

\begin{enumerate}
  \item We extend the existing outsourcing FI mining work to uncertain area w.r.t. the two popular UFI definition criteria and the approximate UFI mining methods.  Specifically, we design the basic checker mechanism verification for expected support UFI definition to verify the random fault, then propose two enhanced schemes based on the private random weights mechanism for the sake of smart cloud attack verification.
    \item Then we explore and design the basic/enhanced verification scheme for PWS based UFI definition, which is able to verify random fault/smart attack. We address the efficiency difficulty of the private weight based enhancement scheme for PWS case, and explore and reduce the computation complexity by eliminating the enumeration of the possible worlds, which raise the efficiency of the enhanced weighted resolution significantly.
  \item As the bridge of the two different UFI  definition,  we further discuss the case of  approximation UFI mining  verification, where we  reduce the verification to the expected support verification and  provide good probabilistic guarantees about the correctness of  verification.
 Finally, we prove the effectiveness and efficiency of the methods proposed in this paper by extensive experiments on synthetic and real datasets.
\end{enumerate}

The rest of the paper is organized as follows: ...

\section{Related Work}
The problem of secure outsourcing the task of data mining with accurate result is emerging recently. While, most of the  existing research  focus on the security and privacy-preserving aspect. \cite{lin2008releasing} proposed a privacy-preserving outsoucing resolution for SVM training and predicting model based on the approximate algorithm on privacy support vectors. \cite{lin2010privacy} utilized the reduce SVM and random transformation technology to deal with the SVM outsoucring problem, to achieve  privacy and accuracy. \cite{wong2007security} addressed the security issues in outsourcing association rule mining. An item mapping and transaction transformation approach was proposed to encrypt a transactional database and to decrypt the mined association rules returned from a service provider.  We can cite more related works on the privacy and security issue on outsourcing database area such as \cite{cao2011privacy,wong2009secure}, etc.

Unlike the privacy-preserving outsourcing problem for data mining, the integrity and verification problem is more difficult and challenging. Though a number of papers have been published on the integrity  verification of the outsourced database model \cite{hacigumus2002providing}, most of them focus on the traditional database queries such as point and range queries \cite{li2007proof,sion2005query}, which return the original qualified tuples or transactions.  For that purpose,  Merkle trees based approaches can be designed to achieve authentication and verification. On the contrary,  in the outsourced data mining model, query results are composed of statistical aggregations (e.g., itemset counts in association rule mining, centroid computation in clustering). The  popular Merkle trees based technique is thus not applicable.

To our best of knowledge,  there are few work on it.  \cite{wong2009audit} proposed  an audit environment, which consists of a database transformation method and a result verification method based on artificial itemset planting (AIP) technique.  But such method will introduce extra mining burden due to the extra fake planted  database, which will influence the verification confidence. Due to this reason, the author propose an alternative, more robust approach for solving the integrity problem, which is based on an aggregate verification mechanism built on inclusion-and-exclusion principle in \cite{Fischer-2006-PhdThesis}. But both approaches are limited in deterministic world, thus not fit the uncertain case.

\section{Basic Definition and Architecture}
\label{sec:definition}
Assume an uncertain database $T$ with $n$ transactions $\{t_1, t_2,\dots,t_n\}$, composed by the items $\{a,b,c,d,e\}$. Specifically, the $i$th transaction $t_i$ with item appearing probability for $\{a_i,b_i,c_i,d_i,e_i\}$ for $\{a,b,c,d,e\}$. Then if $a_i=0$, the item a doesn't appear in $t_i$; similarly,  the item exists definitely in $t_i$ when $a_i=1$. It is easy to see that  once $\{a_i,b_i,c_i,d_i,e_i\}=0\ or\ 1, \forall 1 \leq i \leq n$, the uncertain  scenario will  turn into the traditional deterministic case. Below we will introduce the definition of the deterministic FI and UFI respectively.

\begin{table}
\begin{center}
\caption{uncertain database with item possibility}
\begin{tabular}{|c|c|}
  \hline
   TID &	Transaction \\
   \hline
   $T_1$ &	$a(a_1)b(b_1)c(c_1)d(d_1)e(e_1)$   \\
 \hline
   $T_2$	& $a(a_2)b(b_2)c(c_2)d(d_2)e(e_2)$	\\
  \hline
  $T_3$&	$a(a_3)b(b_3)c(c_3)d(d_3)e(e_3)$\\
\hline
 $\cdots$ & $\cdots$ $\cdots$ \\
\hline
  $T_n$ &	$a(a_n)b(b_n)c(c_n)d(d_n)e(e_n)$\\
  \hline
\end{tabular}
\end{center}
\end{table}

\subsection{Deterministic  Support-based Frequent  Itemset}
\begin{definition}[{Deterministic FI }]
\label{def:derFI}
Given a minimum  support ratio $min\_sup$, an itemset $X$ is an deterministic support-based frequent itemset if and only if the  support exceeds the minimal support threshold
\begin{equation}
\label{equ:deterFI}
  sup(X)=\sum_{i=1}^n{I_i(X)} \geq n  \cdot min\_sup
\end{equation}
where  $X$ can be an item or composite itemset, and the  function $I_i(X)$ indicates whether transaction $t_i$ contains such itemset $X$.

\begin{equation*}
I_i(X)=\begin{cases}
1 & \text{if}\ X \subseteq t_i \\
0 & \text{otherwise}
\end{cases}
\end{equation*}
\end{definition}

As illustrated before, though the deterministic FI definition is concise and has considerable significant works such as Apriori \cite{agrawal1994fast}, FP-growth \cite{han2000mining} and H-Mine \cite{pei2001h}, it is limited to real life applications.

Below we introduce the two semantic definitions of UFI. The first is  Expected Support-based UFI definition, then comes the Probabilistic World Semantic(PWS) based definition.

\subsection{Uncertain Frequent Itemset(UFI) Definition}
\begin{definition}[Expected Support-based UFI]
\label{def:expUFI}
Given a minimum expected support ratio $min\_esup$, an itemset $X$ is an expected support-based frequent itemset if and only if the expected support
\begin{equation}
\label{equ:expUFI}
  esup(X)=\sum_{i=1}^n{p_i(X)} \geq n  \cdot min\_esup
\end{equation}
where $p_i(X)$ indicated the possibility of the item existence in transaction $t_i$.
Under  the popular assumption  that existence of different items is statistical independent \cite{chui2007mining,bernecker2009probabilistic,calders2010approximation},   the  probability  of   itemset $X$ can be  obtained  by  simply  multiplying  the  individual  item  probabilities $p_i(x)$ as in Equation~\eqref{equ:itemsetProb} below:
\begin{equation*}
\label{equ:itemsetProb}
  p_i(X)=\prod_{x \in X}p_i(x)
\end{equation*}
\end{definition}

We note that the definition in Equation~\eqref{equ:expUFI} can be seen as a natural generalization of the deterministic FI definition in Equation~\eqref{equ:deterFI} with possibility function $p_i(X)$ in Equation~\eqref{equ:itemsetProb} instead of inclusion function $I_i(X)$.

Though the definition of Expected Support-based FI uses the expectation to measure the uncertainty, which is a simply extension of the definition of the frequent itemset in deterministic data and  is known as an important statistic, it cannot show the complete probability distribution. Therefore, the other semantic based definition is offerer in Probabilistic World Semantic(PWS).  Furthermore, most prior researches believe that the two definitions should be studied respectively in \cite{sun2010mining,bernecker2009probabilistic,wang2010accelerating}, we first give the definition of Possible World below.

\begin{definition}[Possible World]
\label{def:possibleworld}
A possible world $\mathcal{W}$ is a deterministic sample or subseteq of uncertain database $T$, whose possible denoted as $P_T(W)$ and obtained by the multiply of the presence and absence possibility according to the item independence assumption.
\begin{equation}
\label{equ:probPWS}
  P_T(W)= \prod_{I_{W}(x,i)=1}p_i(x)\cdot \prod_{I_{W}(x,i)=0}(1-p_i(x))
\end{equation}

where function $I_{W}(x,i)$ indicates whether the item $x$ in transaction $t_i$  is contained in world $\mathcal{W}$.

\begin{equation*}
I_W(x,i)=\begin{cases}
1 & \text{if}\ x \subseteq t_i \in \mathcal{W} \\
0 & \text{otherwise}
\end{cases}
\end{equation*}
\end{definition}
Now we use notation $W_T$ to represent the possible world set of uncertain database $T$ and introduce the PWS based UFI definition as follows.

\begin{definition}[Probabilistic World Semantic(PWS) based UFI]
\label{def:pwsUFI}
Given a minimum support ratio $min\_sup$, and a probabilistic frequent threshold $pft$, an itemset $X$ is a probabilistic frequent itemset if and only if its frequent probability, obtained by the sum of qualified possible worlds possibilities, denoted as $Pr(X)$ satisfies \footnote{we use the \textbf{Iversion} symbol in book of  Concrete Mathematics for the elegancy of equation }:
\begin{equation}
\begin{split}
    Pr(X)&=Pr(sup(X) \geq min\_sup)\\
         &=\sum_{W \in W_T} P_W(T) \cdot [sup(X,W)\geq min\_sup]\geq pft
\end{split}
\end{equation}

\end{definition}

\subsection{Architecture of verification framework}
\begin{figure*}
  \centering
  \subfigure[pre-mining architecture based on extra artificial database]{
    \label{fig:architecture:a} 
    \includegraphics[width=.4\textwidth]{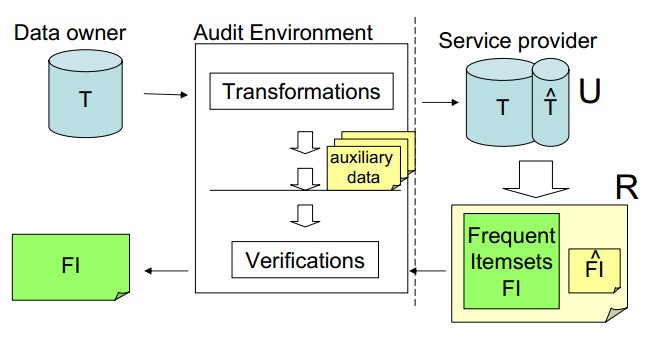}}
  \subfigure[post-mining architecture with checker mechanism ]{
    \label{fig:architecture:b} 
    \includegraphics[width=.35\textwidth]{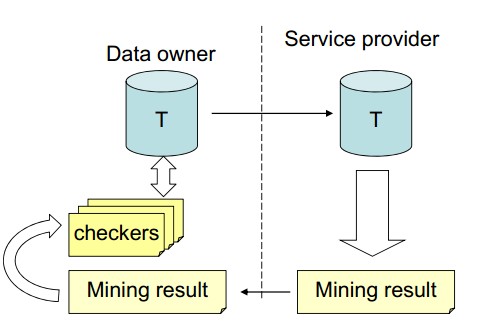}}
  \caption{Two possible architectures for integrity verification of outsourcing frequent itemset mining}
  \label{fig:architecture} 
\end{figure*}
In this section, we will address the architecture of the verification framework. We summarize the two possible  choices of verification resolution in this paper as in Figure~\ref{fig:architecture}. Now we illustrate the two choices respectively according to the .

\textbf{The pre-mining architecture.} As in Figure~\ref{fig:architecture:a}, the  audit and verification environment consists of a database transformation method and a result verification method based on artificial fake database plant technique.  Because the construction of the environment is before the actual computation of FIs, thus it belongs to the pre-mining architecture according to \cite{Fischer-2006-PhdThesis,wong2009audit}.
Intuitively, the method following Figure~\ref{fig:architecture:a} will introduce extra computation burden due to the extra fake  database, which will be the crux and influence the verification confidence.

\textbf{The post-mining architecture.} Due to this reason, the other verification choice in Figure~\ref{fig:architecture:b} can be utilized to avoid such bottleneck, where the architecture use some specific checker mechanism to achieve verification. Obviously, the checker mechanism preforms right after the mining phase, thus it can be seen as the post-mining architecture. Specifically, for FI mining, \cite{Fischer-2006-PhdThesis} proposed a representative   based on the inclusion-and-exclusion principle.

Here it is worth mentioning that the two choices of integrity verification are both applicable to deterministic and uncertain data mining cases, thus general to all algorithms. But there is no paper working on the uncertain scenario and problems. Specifically, for FIs mining, only exist \cite{Fischer-2006-PhdThesis,wong2009audit}  for deterministic discussion. In this paper, we will follow the way of post-mining mechanism based on checker mechanism and explore the verification resolution for UFI mining.

\section{Deterministic FI Mining Verification}
\label{sec:deterministic}
In order to verify every returned frequent itemset, we will  create a count checker for each maximal itemset $X$(which is frequent itself, but all
its supersets are not frequent). This method is expected to be efficient  since the number of maximal itemsets $X$ is expected to be much smaller than the size of itemsets number. In this section, we will illustrate the existing FI verification scheme in \cite{Fischer-2006-PhdThesis}.
\subsection{Basic count checker verification}
\begin{definition}[{Count checker }]
\label{def:countchecker}
Let X be an itemset, referred to as the checkset. The count checker of X, denoted by $X^*.cnt$, is defined as the total support count of its non-empty subsets, that is,  $\mid \{t \mid t \in T \ and \ t \cap X \neq \varnothing \}\mid$.
\end{definition}

In the deterministic FI mining  scenario, for single item FI $X$ such as  $B, C$, we get $X^*.cnt=X.cnt=sup(X)$, thus $B^*.cnt=sup(B)$, $C^*.cnt=sup(C)$. For example, if the database $T$ contains 5 transactions, that is,  $T=\{ABC, AB, ABD, CD, AD\}$, then $\{B\}^*.cnt=3$, $\{C\}^*.cnt=2$. The complex itemset case can be computed by the one scan of the database $T$ according to the definition in Definition~\ref{def:countchecker}. Hence $\{BC\}^*.cnt = 4$ in the example above.

The other way is to compute the count checker of complex itemset is by inclusion-exclusion principe. Take $X=BC$ as an example, we define  $sup(B \cup C)$ as the support count of the any subset of $BC$, which can be thus  $\{BC\}^*.cnt=sup(B \cup C)= sup(C) + sup(B)- sup(BC)=2+3-1=4$, thus  the verification scheme can be done by one scan complexity $\mathcal{O}(n)$ as summarized in Theorem~\ref{the:basicdeterverify} below.
\begin{theorem}[{deterministic FI  checker verification}]
 \label{the:basicdeterverify}
In deterministic database outsourcing mining,  the verification based on  inclusion-exclusion principe count checker is effective, that is to say,
$$X^*.cnt=\sum_{i=1}^{\mid X \mid}\sum_{Y \subseteq X, \mid Y \mid =i}(-1)^{i+1}sup(Y)$$
\end{theorem}
\begin{proof}
The proof of Theorem~\ref{the:basicdeterverify} can be achieved by the inclusion-exclusion principle on the subset of $X$, thus we neglect the details here.
\end{proof}

Consider the previous example, where $\{BC\}^*.cnt = 4$ is obtained by scanning the database and assume that cloud returns $sup(B)=3$, $sup(C)=2$ and $sup(BC)=2$. Computing  $\{BC\}^*.cnt=sup(B \cup C)= sup(B) + sup(C)- sup(BC)=3+2-2=3$ using according to the equation above, the result  is different from the real checker $\{BC\}^*.cnt = 4$. As a consequence, we can know that cloud party is not honest since the verification detects such fraud. We note that the basic checker verification mechanism is can detect the random fault/stupid cloud attack, that is to say, if an honest/stupid cloud returns the result, the basic verification scheme can detect it. However, once a smart cloud aware of this mechanism and returns  $sup(B)=3$, $sup(C)=2$ and $sup(BC)=1$, then the detection will fail as   $\{BC\}^*.cnt=sup(B \cup C)= sup(B) + sup(C)- sup(BC)=3+2-1=4$, which will make the scheme unsetting.

\textbf{\emph{Obviously,  the  verification method in Theorem~\ref{the:basicdeterverify} will not bring into any extra computation cost for the cloud party, therefore, it is efficient and can detect random fault/stupid cloud attack abnormal.
}}

\subsection{Attack-Resistant Checkers with Private Weights}

. In order to prevent such attacks, the checker should include private information that is only known to the database owner, so that cloud party cannot reproduce the same aggregate result without knowing the private parameters.

The parameters impose a virtual database transformation; the data owner applies the verification on the transformed database $T'$ instead of the original $T$. $T'$ is generated by replicating each transaction $t \in T$ by a weight $w_t$, which is determined by the private parameters. We stress that the transformation is never actually performed or materialized; instead the mechanism only assumes that this transformation has been done but uses the mining result from cloud service provider and the original database $T$ to complete the verification.

We now discuss the details of the transformation and the adapted checking mechanism. We assign a random weight $w_x$ on each item $x$. The set of item
weights constitute the private information held by data owner itself. Now, we assume that the database is not $T$, but a transformed database $T'$, such that for each transaction $t \in T$, there are $w_t$ transactions in $T'$ , where $w_t = \prod_{x \in t} w_x$. For example, if  $w_A = 6$, $w_B = 4$, transaction $\{AB\}$ has a weight of $w_A \times  w_B = 24$, thus it appears 24 times in $T'$.

For example, if the original database   $T=\{AB, A, B\}$, then $sup(A)=sup(B)=2$, $sup(AB)=1$, $\{AB\}^*.cnt=3$. According to the private mechanism, without loss of generality, we choose the random private weights $w_A=2$, $w_B=3$, then after such virtual transformation, we can note  $T \to T': \{AB \times 6, A\times 2, B \times 3\}$. Now we can compute  the new support counts in the transformed virtual database $T'$ as below:
\begin{equation*}
\begin{cases}
sup(A)_{T'}=w_Aw_Bsup(AB)+w_A(sup(A)-sup(AB))=8 \\
sup(B)_{T'}=w_Aw_Bsup(AB)+w_B(sup(B)-sup(AB))=9\\
sup(AB)_{T'}=w_Aw_B sup(AB)=6\\
\end{cases}
\end{equation*}

And compute the  new count checker $\{AB\}^*.cnt=w_A+w_B+w_Aw_B=11$ by one scan of the database $T'$. Then the verification equation can be established as
\begin{align*}
\begin{split}
\{AB\}^*.cnt&= sup(A)_{T'}+sup(B)_{T'}-sup(AB)_{T'}\\
&=w_Asup(A)+w_Bsup(B)+(w_Aw_B\\
& \quad -w_A-w_B)sup(AB)\\
&=8+9-6=11
    \end{split}&
\end{align*}

Now we can see that the verification don't require the actual transformation from $T$ to $T'$,  instead  we can  assumes that done but utilize the support count from the returned database $T$ to complete the new checker verification above.

In fact, for any frequent itemset $X$,  we can  complete the  computation of $X^*.cnt_{T'}$ on the transformed database $T'$ with one scan. Furthermore,s we can  follow a similar procedure as Theorem~\ref{the:basicdeterverify} to complete the verification based on the inclusion-exclusion principle in Theorem~\ref{the:privatedeterverify}.

\begin{theorem}[{deterministic private weights based checker verification}]
 \label{the:privatedeterverify}
In deterministic database outsourcing mining,  the verification based on  inclusion-exclusion principe private  count checker is effective, that is to say,
$$X^*.cnt_{T'}=\sum_{Y \subseteq X, \mid Y \mid > 0} (\sum_{Z \subseteq Y, \mid Z \mid >0}(-1)^{\mid Y-Z\mid} \prod_{z\in Z}w_z)sup(Y)_T$$
\end{theorem}

\begin{proof}
  The proof of Theorem~\ref{the:privatedeterverify} can be referenced to \cite{Fischer-2006-PhdThesis}.
\end{proof}

\textbf{\emph{Similarly, the private weights based verification method in Theorem~\ref{the:privatedeterverify} also will not bring into any extra computation cost for the cloud party and satisfy somewhat good efficiency.
}}
Therefore, we can get the conclusion in this section as \textbf{\emph{Remarks: in deterministic outsourcing FI mining scenario, the basic checker mechanism and its enhanced  private weights based resolution can both be built effectively.
}}

Below we will explore the FI outsourcing mining problem in uncertain scenario, which extends the existing work and is still a blank to our best of knowledge.  As illustrated in Section~\ref{sec:definition}, there exist two different UFI definition, Expected Support-based UFI and PWS based UFI, thus we will explore and address the outsourcing UFI  mining verification  in the later Section~\ref{sec:exp} and Section~\ref{sec:pws} repectively.

\section{ Expected Support-based UFI Verification}
\label{sec:exp}
Due to the deterministic FI mining can be seen as a special case of UFI mining, that is, when the appearance probability  $\{a_i,b_i,c_i,d_i,e_i\}=0\ or\ 1, 1 \leq i \leq n$, the scenario turns into the traditional deterministic database case above. In this section we will follow the Expected Support-based  UFI definition.

\subsection{Basic count checker verification}
Below  we first extend the count checker into the uncertain scenario and propose the \emph{Expected Support-based count checker}.
\begin{definition}[{Expected Support-based count checker}]
\label{def:expcountchecker}
Let X be an itemset, referred to as the checkset. The Expected Support-based count checker of X, denoted by $X^*.ecnt$, is defined as the total probability of its non-empty subset, that is,  $\sum_{t} \{p(t) \mid t \in T \ and \ t \cap X \neq \varnothing \}$.
\end{definition}

Similar to the deterministic FI mining  scenario, for single item itemset $X$ such as  $A, B$, we get $X^*.ecnt=X.ecnt=esup(X)$, thus $A^*.ecnt=esup(A)=\sum_{i=1}^na_i$, $B^*.ecnt=esup(B)=\sum_{i=1}^{n}b_i$.

The complex itemset case can be computed by one scan of the database $T$ according to  Definition~\ref{def:expcountchecker}, specifically, according to the  possibility of the reverse event
\begin{align*}
\begin{split}
\{AB\}^*.ecnt&=esup(A \cup B)=esup(\overline{{A} \cap {B}})\\
&=\sum_{i=1}^n (1-(1-a_i)(1-b_i))\\
&= \sum_{i=1}^n(a_i+b_i-a_ib_i)\\
    \end{split}&
\end{align*}

Similar to the case of deterministic FI mining,  the other way  to compute the checker of $\{AB\}^*.ecnt$ is by inclusion-exclusion principe. Thus the verification can be preformed by one scan complexity $\mathcal{O}(n)$, summarized in Theorem~\ref{the:expbasicverify} below.

\begin{theorem}[basic checker verification for Expected Support-based FI]
 \label{the:expbasicverify}
The verification and checker based on inclusion-exclusion principe can fit the Expected Support-based FI definition, specifically,
$$X^*.ecnt=\sum_{i=1}^{\mid X \mid}\sum_{Y \subseteq X, \mid Y \mid =i}(-1)^{i+1}esup(Y)$$
\end{theorem}
\begin{proof}
  The proof of Theorem~\ref{the:expbasicverify} is similar to Theorem~\ref{the:basicdeterverify}.The base is that for all single item itemset $X$, $X^*.ecnt=X.ecnt=esup(X)$, furthermore, the general complex $X$ can be induced by the inclusion-exclusion principle.
\end{proof}

\textbf{\emph{Obviously, we noted that the basic checker verification for Expected Support-based FI definition in Theorem~\ref{the:expbasicverify} will not bring into any extra computation cost for the cloud party, therefore, it is efficient.
}}

For a simple example, $X=AB$, according to the inclusion-exclusion principle,    we can establish the verification equation:
$\{AB\}^*.ecnt= esup(A) + esup(B) - esup(AB)$.
Here we build a simple database with two transactions and two items $A$ and $B$ in the table below  as $T$ and  assign them different random existence possibilities. According to the definition

\begin{center}
  \begin{tabular}{|c|c|}
  \hline
   TID &	Transaction \\
   \hline
   $T_1$ &	$A(0.5)B(0.6)$   \\
 \hline
   $T_2$	& $A(0.4)B(0.5)$	\\
   \hline
\end{tabular}
\end{center}
 of Expected Support-based FI in Equation~\eqref{equ:expUFI}. We can get the expected support count from the cloud party with  $esup(A)=\sum_{i=1}^na_i=0.9$, $esup(B)=\sum_{i=1}^nb_i=1.1$, $esup(AB)=\sum_{i=1}^na_ib_i=0.5$.  According to the inclusion-exclusion principle, we can compute
\begin{align*}
\begin{split}
    \{AB\}^*.ecnt&= esup(A) + esup(B) - esup(AB)\\
    &=0.9+1.1-0.5=1.5
\end{split}&
\end{align*}

Then according to the one scan computation method, we can recompute and get the verification below
\begin{align*}
\begin{split}
\{AB\}^*.ecnt&=\sum_{i=1}^2(a_i+b_i-a_ib_i)\\
            &=(1-0.5\times 0.4)+(1-0.6\times 0.5)\\
            &=0.8+0.7=1.5\\
            &=esup(A) + esup(B) - esup(AB)
\end{split}&
\end{align*}

\subsection{Enhanced Checkers with Private Weights}
Similar to the discussion in Section~\ref{sec:deterministic}, the basic checker mechanism maybe suffer the smart cloud party attack which will pass the  inclusion-exclusion principle based verification. It is worth noting that the similar problem goes with the basic scheme above for expected-support based UFI computation. Thus we will explore the similar  privacy weighted parameters enhanced scheme  below.

Similarly, we will aim to design such private random parameters based  virtual database transformation. In other words, the data owner applies the verification on the transformed database $T'$ instead of the original $T$.

\subsubsection{\textbf{first scheme}}
First we follow the similar way to set the random private item weights. Then unlike the case of deterministic FI mining, we generate  the $T'$  by item existence possibility scaling operation, specifically,  multiplying the each item existence possibility by weight $w_t$, which is determined by the product of the  private item  weights as $\prod_{x \in t}w_x$. Here we note that we choose the multiplication operation as the virtual transformation operation rather than the replication, which can meet the property of uncertain database well. What's more, the deterministic scenario replication transformation is just a special case of  possibility scaling on integer  possibility.

We now discuss the details of the transformation and the adapted checking mechanism. We assign a random weight $w_x$ on each item $x$.

 Take the previous example, if  $w_A = 0.4$, $w_B = 0.5$, then transaction  $AB$ will  scale up with $M=\prod_{x\in AB}w_x=w_Aw_B=0.4\times 0.5=0.20$, thus the original uncertain database $T$ can be transformed into the virtual form $T'$ below:

\begin{center}
  \begin{tabular}{|c|c|}
  \hline
   TID &	Transaction \\
   \hline
   $T_1$ &	$A(0.5M=0.10)B(0.6M=0.12)$   \\
 \hline
   $T_2$	& $A(0.4M=0.08)B(0.5M=0.10)$	\\
   \hline
\end{tabular}
\end{center}

 Now we can compute  the new expected support counts in the transformed virtual database $T'$ as below:
\begin{equation*}
\begin{cases}
esup(A)_{T'}=0.5M+0.4M=0.9M=Mesup(A)_T \\
esup(B)_{T'}=0.6M+0.5M=1.1M=Mesup(B)_T\\
esup(AB)_{T'}=0.3M^2+0.2M^2=0.5M^2=M^2esup(AB)_T\\
\end{cases}
\end{equation*}

Then we compute the  new count checker by inclusion-exclusion principle
\begin{align*}
\begin{split}
    \{AB\}^*.ecnt&=Mesup(A)_T+Mesup(B)_T -M^2esup(AB)_T\\
    &=0.18+0.22-0.02=0.38\\
\end{split}&
\end{align*}
Through one scan of the database $T'$, the verification equation can be established as

Then according to the one scan computation method, we can recompute and get the verification below
\begin{align*}
\begin{split}
\{AB\}^*.ecnt&=\sum_{i=1}^2(a_i+b_i-a_ib_i)\\
            &=(1-0.9\times 0.88)+(1-0.92\times 0.9)\\
            &=0.38
\end{split}&
\end{align*}

We note that the example above, the transformed existence possibilities in $T'$ for all items are all below 1 due to the choice of $w_A = 0.4$, $w_B = 0.5$. In fact, even when we choose the large such verification will still succeed because the verification equation below isn't related with constant 1.
\begin{align*}
\begin{split}
    \{AB\}^*.ecnt&= esup(A) + esup(B) - esup(AB)\\
    &=\sum_{i=1}^n(a_i+b_i-a_ib_i)
\end{split}&
\end{align*}

Thus $w_A = 4$, $w_B = 5$ which will result in $M=20$ and
  \begin{center}
  \begin{tabular}{|c|c|}
  \hline
   TID &	Transaction \\
   \hline
   $T_1$ &	$A(0.5M=10)B(0.6M=12)$   \\
 \hline
   $T_2$	& $A(0.4M=8)B(0.5M=10)$	\\
   \hline
\end{tabular}
\end{center}
will not break such effectiveness. And the privacy weights choice and virtual transformation can be relaxed and easy.

Thus the verification can be preformed by one scan complexity $\mathcal{O}(n)$, summarized in Theorem~\ref{the:expfirstverify} below.

\begin{theorem}[First private weights based checker verification for Expected Support-based FI]
 \label{the:expfirstverify}
The verification and checker based on inclusion-exclusion principe can fit the Expected Support-based FI definition, specifically,
$$X^*.ecnt=\sum_{i=1}^{\mid X \mid}\sum_{Y \subseteq X, \mid Y \mid =i}(-1)^{i+1}M^iesup(Y)_T$$
where $M=\prod_{x\in I}w_x$, $I$ is the item set.
\end{theorem}
\begin{proof}
The proof of Theorem~\ref{the:expbasicverify} can be completed below. According to the first private weights based verification scheme, the base is that  $\forall  |X|=1$, $X^*.ecnt=X.ecnt_{T'}=M\cdot esup(X)_T$, then  similarly, $\forall |X|=i$, $X.ecnt_{T'}=M^i \cdot esup(X)_T$, Thus the general case $X^*.ecnt$ can be induced by the inclusion-exclusion principle.
\end{proof}

In the example above, $I=\{A,B\}$, $M=w_Aw_B=0.2$,  thus equation $\{AB\}^*.ecnt=Mesup(A)_T+Mesup(B)_T -M^2esup(AB)_T$ can be deduced by Theorem~\ref{the:expfirstverify}.

\subsubsection{\textbf{Second scheme}}
The private weights based checker mechanism is based on multiplying the product of the  private item  weights  $M=\prod_{x \in t}w_x$ in the transaction . The idea is  extended from the deterministic scenario naturally. Now we propose another multiplication based transformation operation on the isolated item possibilities, instead of the whole product of the item weights in the transaction.

Take the previous example, if  $w_A = 0.4$, $w_B = 0.5$, now the original uncertain database $T$ can be transformed into the virtual form $T'$ via multiplying the item possibility by  corresponding private weights instead of the whole weight product $M$   below:

\begin{center}
  \begin{tabular}{|c|c|}
  \hline
   TID &	Transaction \\
   \hline
   $T_1$ &	$A(0.5w_A=0.20)B(0.6w_B=0.30)$   \\
 \hline
   $T_2$	& $A(0.4w_A=0.16)B(0.5w_B=0.25)$	\\
   \hline
\end{tabular}
\end{center}

 Now we can compute  the new expected support counts in the transformed virtual database $T'$ as below:
\begin{equation*}
\begin{cases}
esup(A)_{T'}&=0.5w_A+0.4w_A=0.9w_A=w_Aesup(A)_T \\
esup(B)_{T'}&=0.6w_B+0.5w_B=1.1w_B=w_Besup(B)_T\\
esup(AB)_{T'}&=0.3w_Aw_B+0.2w_Aw_B=0.5w_Aw_B\\
 &=w_Aw_Besup(AB)_T\\
\end{cases}
\end{equation*}

Then we compute the  new count checker by inclusion-exclusion principle
\begin{align*}
\begin{split}
    \{AB\}^*.ecnt&=w_Aesup(A)_T+w_Besup(B)_T -w_Aw_Besup(AB)_T\\
    &=0.36+0.55-0.10=0.81\\
\end{split}&
\end{align*}

Then according to the one scan computation method, we can recompute and get the verification below
\begin{align*}
\begin{split}
\{AB\}^*.ecnt&=\sum_{i=1}^2(a_i+b_i-a_ib_i)\\
            &=(1-0.8\times 0.7)+(1-0.84\times 0.75)\\
            &=0.81
\end{split}&
\end{align*}

Thus the verification can be preformed by one scan complexity $\mathcal{O}(n)$, summarized in Theorem~\ref{the:expsecondverify} below.

\begin{theorem}[Second private weights based checker verification for Expected Support-based FI]
 \label{the:expsecondverify}
The verification and checker based on inclusion-exclusion principe can fit the Expected Support-based FI definition, specifically,
$$X^*.ecnt=\sum_{i=1}^{\mid X \mid}\sum_{ Y \subseteq X, \mid Y \mid =i}(-1)^{i+1}\prod_{z\in Y} w_zesup(Y)_T$$
where $w_z$ is the privacy weight of item $z\in I$, $I$ is the item set.
\end{theorem}
\begin{proof}
The proof of Theorem~\ref{the:expsecondverify} can be completed as follows. According to the second private weights based verification scheme, the base is that  $\forall  |X|=1$, $X^*.ecnt=X.ecnt_{T'}= w_X esup(X)_T$, then  similarly, $\forall |X|=i$, $X.ecnt_{T'}=\prod_{z\in X}w_z \cdot esup(X)_T$, Thus the general case $X^*.ecnt$ can be induced by the inclusion-exclusion principle in Theorem~\ref{the:expsecondverify}.
\end{proof}

In the example above, $I=\{A,B\}$, $w_A=0.4$, $w_B=0.5$, thus  $ \{AB\}^*.ecnt=w_Aesup(A)_T+w_Besup(B)_T -w_Aw_Besup(AB)_T$  can be deduced by Theorem~\ref{the:expsecondverify}.

\textbf{\emph{Similarly, the two private weights based verification method in Theorem~\ref{the:expfirstverify}  and Theorem~\ref{the:expsecondverify} both will not bring into any extra computation cost for cloud party and satisfy good efficiency.
}}
Therefore, we can get the conclusion in this section as \textbf{\emph{Remarks: in expected support-based outsourcing FI mining scenario, there exists effective basic checker mechanism, and its enhanced private weights based resolutions can also be built effectively.
}}


\section{ Probabilistic PWS based UFI Verification}
\label{sec:pws}
Though expected support-based UFI definition utilizes the support expectation to measure the uncertainty of items, which is a simply extension of the  deterministic case and can be finished with somewhat low cost, it cannot show the complete probability distribution.
In order to overcome this problem, Probabilistic World Semantic(PWS) based UFI definition was proposed. However, the characteristics of PWS,  the  combinational relation between the transactions, differs from the previous  simple support sum  and  expected support sum of the FIs when deterministic FI and expected support-based UFI definition.
Now we explore whether we can obtain similar potential checker mechanism with such PWS based UFI definition. Below we can easily get the similar PWS based count checker definition below.

\subsection{Basic count checker verification}

\begin{definition}[{PWS based count checker}]
\label{def:probcountchecker}
Let X be an itemset, referred to as the checkset. The PWS based count checker of X, denoted by $X^*.pcnt$, is defined as the total probability of its non-empty qualified possible worlds w.r.t. to the minimum support ratio min\_sup,
\begin{align}
\begin{split}
X^*.pcnt&=\sum_t Pr[sup(t) \geq \delta \mid t \subset X, t \neq \varnothing])\\
    \end{split}&
\end{align}
where threshold $\delta=n \cdot min\_sup$.
\end{definition}

According to Definition~\ref{def:probcountchecker}, for all the single item itemset $X$ such as  $A, B$, we can get
\begin{equation*}
\begin{cases}
A^*.pcnt=Pr[sup(A) \geq  \delta] \\
B^*.pcnt=Pr[sup(B) \geq  \delta ]\\
\end{cases}
\end{equation*}

In fact, for the general itemset $X$, $X^*.pcnt=X.pcnt=Pr(X)$. Take simple set $AB$ as example,
$$AB.pcnt=Pr[sup(AB) \geq  \delta ]$$

Below we explore the verification of $\{AB\}^*.pcnt$ according to Definition~\ref{def:probcountchecker}.
  \begin{align}
  \label{equ:probbasis}
\begin{split}
\{AB\}^*.pcnt&=Pr[sup(A \cup B) \geq \delta]\\
    &= Pr[sup(A) \geq  \delta \lor sup(B) \geq  \delta ]\\
    \end{split}&
\end{align}

 Now we can deduce from the Equation\eqref{equ:probbasis}  with inclusion-exclusion principle,
 \begin{align}
 \begin{split}
 \label{equ:probinclude}
\{AB\}^*.pcnt&= Pr[sup(A) \geq  \delta]+Pr[sup(B) \geq  \delta ]\\
 & \qquad -Pr[sup(AB) \geq  \delta ]\\
    &=A.pcnt+B.pcnt-AB.pcnt
\end{split}&
\end{align}

We note that the probability of the parts $A.pcnt$, $B.pcnt$, $AB.pcnt$ in Equation~\eqref{equ:probinclude} will be computed and returned by cloud party.
Meanwhile, in the other way, we can  find that $\{AB\}^*.pcnt$ can be obtained by the inverse event probability as in Equation~\eqref{equ:probinverse},
\begin{align}
\begin{split}
\label{equ:probinverse}
\{AB\}^*.pcnt&=Pr[sup(A) \geq  \delta \lor sup(B) \geq  \delta ]\\
    &=  1-Pr[sup(A) <  \delta \land sup(B) <  \delta ]\\
    &=  1-Pr[0  <sup(A) <  \delta \land 0<sup(B) <  \delta ]\\
    &\qquad    -Pr[sup(A)=0 \land sup(B)=0]\\
    \end{split}&
\end{align}

According the similar analysis before, the part $Pr[sup(A)=0\land sup(B)=0]$ can be computed by by one scan complexity $\mathcal{O}(n)$ as $\prod_{i=1}^n(1-a_i)(1-b_i)$.
Then when we combine Equation \eqref{equ:probinclude} and \eqref{equ:probinverse}, we can get
\begin{align}
\begin{split}
\label{equ:pcnt}
\{AB\}^*.pcnt&=A.pcnt+B.pcnt-AB.pcnt\\
    &= 1-Pr[0  <sup(A) <  \delta \land 0<sup(B) <  \delta ]\\
    &\quad -\prod_{i=1}^n (1-a_i)(1-b_i)\\
    \end{split}
\end{align}

Thus once the  part $\lambda=Pr[0  <sup(A) <  \delta\land 0<sup(B) <  \delta ]$ can be  computed and returned along with the frequent itemset probabilities such as $A.pcnt$, $B.pcnt$, $AB.pcnt$ by some specific PWS based UFI mining algorithm, $\{AB\}^*.pcnt$ can be verified  by one scan of the database as in Equation~\eqref{equ:pcnt}. Then we can build the basic checker mechanism for the general itemset $X$ in Theorem~\ref{the:probbasicvefify} below.

\begin{theorem}[{PWS based Probability basic count checker verification}]
 \label{the:probbasicvefify}
In UFI outsourcing mining, the verification is effective with the basic count checker when PWS based Probability FI definition. Specifically, for general itemset $X$,
\begin{align}
\begin{split}
\label{equ:pcnt}
\{X\}^*.pcnt&=\sum_{i=1}^{\mid X \mid}\sum_{Y \subseteq X, \mid Y \mid =i}(-1)^{i+1}Y.pcnt\\
    &=Pr[\bigwedge_{z\in X}  0  <sup(z) <  \delta ]\\
    &= 1-Pr[\bigwedge_{z\in X}  0  <sup(z) <  \delta ]\\
    &\quad -\prod_{z\in X, i=1}^n (1-z_i)\\
    \end{split}
\end{align}

\end{theorem}
\begin{proof}
  The evidence is obvious. On one side, the general $\{X\}^*.pcnt$ can be computed by the inclusion-exclusion principle as $\sum_{i=1}^{\mid X \mid}\sum_{Y \subseteq X, \mid Y \mid =i}(-1)^{i+1}Y.pcnt$. On the other side, it can be represented by the inverse probability as $1-Pr[\bigwedge_{z\in X}  0  <sup(z) <  \delta ] -\prod_{z\in X, i=1}^n (1-z_i)$, which is the general form of Equation~\eqref{equ:pcnt}.
\end{proof}

Now we note that part $Pr[\bigwedge_{z\in X}  0  <sup(z) <  \delta ]$ is somewhat confusing and difficult for us to compute. Thus here we should answer two questions:
\begin{enumerate}
  \item \textbf{whether the computation of $Pr[\bigwedge_{z\in X}  0  <sup(z) <  \delta ]$ is necessary; }
  \item \textbf{If so, how can we design such a efficient way for cloud party to resolve it efficiently.}
\end{enumerate}

First, we give the example with $X=AB$, $\delta=1,2$  in Figure~\ref{fig:probnecessity} to illustrate  the necessity of $Pr[\bigwedge_{z\in X}  0  <sup(z) <  \delta ]$ with the possible worlds possibilities points. In Figure~\ref{fig:probnecessity}(a)(b), the horizontal and vertical axes represent the support of the item of $B$ and $A$, then each coordinate point with one pair of support value represents the sum of the possible world possibilities, e.g., the point $(0,0)$ represent $Pr[sup(A)=0 \land sup(B)=0]$, and $(x,y)$ represents $Pr[sup(A)=y \land sup(B)=x]$.    As in Figure~\ref{fig:probnecessity}(a), the frequent threshold $\delta=1$, then $Pr[sup(A)=0 \land sup(B)=0]$ indicated in red, the part $Pr[sup(A) \geq  \delta \lor sup(B) \geq  \delta ]$ can be indicated in the set of green points, we note that when $\delta=1$ $Pr[0  <sup(A) <  \delta \land 0<sup(B) <  \delta ]=0$, thus not existing in Figure~\ref{fig:probnecessity}.  However,  when threshold $\delta =2 \neq 1$ in Figure~\ref{fig:probnecessity}(b), $Pr[0  <sup(A) <  \delta \land 0<sup(B) <  \delta ]=0$  can be indicated in coordinate points in yellow.

\begin{figure}[htb]
  \centering
  \includegraphics[width=.4\textwidth]{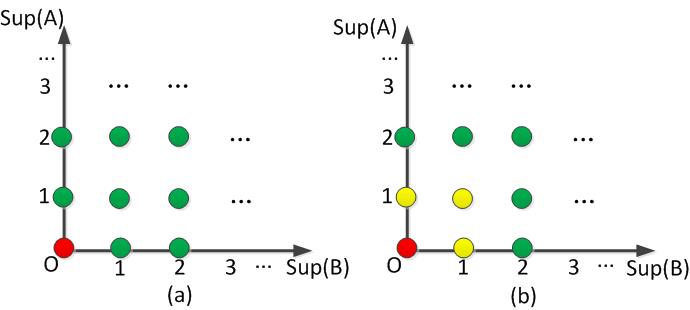}\\
  \caption{(a)$\lambda=0$ when $\delta=1$; (b)$\lambda \neq 0$ when $\delta=2$}\label{fig:probnecessity}
\end{figure}

Therefore, we can get the conclusion that once threshold $\delta > 1$, computation of $Pr[\bigwedge_{z\in X}  0  <sup(z) <  \delta ]$ in Therorem~\ref{the:probbasicvefify} is necessary. In other words, if the cloud can't complete the computation of $Pr[\bigwedge_{z\in X}  0  <sup(z) <  \delta ]$, the verification in Equation~\eqref{equ:pcnt} can't be established.

Then we discuss the possible efficient resolution of the computation on $Pr[\bigwedge_{z\in X}  0  <sup(z) <  \delta ]$. To our best of knowledge, existing popular efficient  PWS based UFI mining methods \cite{bernecker2009probabilistic,sun2010mining} all utilize the pruning mechanism to avoid the computation the probabilities of the unfrequent UFIs, in other words,  they only care and compute the probability value of the UFIs, not the unfrequent itemsets  in the part $Pr[\bigwedge_{z\in X}  0  <sup(z) <  \delta ]$. Thus we should design some new extra method for its computation.
We explore and analyze the problem from the basis,  $X$ as a single item composed itemset.

\textbf{\emph{Basis. }} $X$ as a single item composed itemset. In this case, the object $Pr[\bigwedge_{z\in X}  0  <sup(z) <  \delta ]$ will be  simplified to $Pr[ 0  <sup(X) <  \delta ]$, which means the total sum  of the unfrequent possible world probabilities containing $X$.  Here we can follow the similar dynamic programming way in UApriori\cite{bernecker2009probabilistic} to design the method. Similar to the notations in \cite{bernecker2009probabilistic}, where the notation $P_{\geq i,j}(X)$  was used to denotes the probability that itemset $X$ appears at least $i$ times among the first $j$ transactions in the given uncertain database with $N$ transactions, here we use notation $P_{< i,j}(X)$ to denotes the probability that itemset $X$ appears less than $i$ times among the first $j$ transactions. Therefore, the recursive relationship is defined as follows:
$$P_{< i,j}(X)=P_{< i-1,j}(X)\cdot P(X \subseteq t_j)+ P_{< i,j}(X)\cdot P(X \nsubseteq t_j)$$
where the boundary case:
\begin{equation*}
\begin{cases}
P_{< i,j}(X)=1, \qquad j<i\\
P_{< 0,j}(X)=0, \qquad 0 \leq j \leq N\\
\end{cases}
\end{equation*}

Therefore, similar to the object $P_{\geq \delta, N}(X)$ in \cite{bernecker2009probabilistic}, here our object can be formalized as $P_{< \delta, N}(X)$. Then we can compute it in the dynamic programming way depicted in Figure~\ref{fig:probdynamic}, where we use the notation $P_{i,j}(X)$ for the simplification of  $P_{< i,j}(X)$ above.
\begin{figure}[htb]
  \centering
  \includegraphics[width=.4\textwidth]{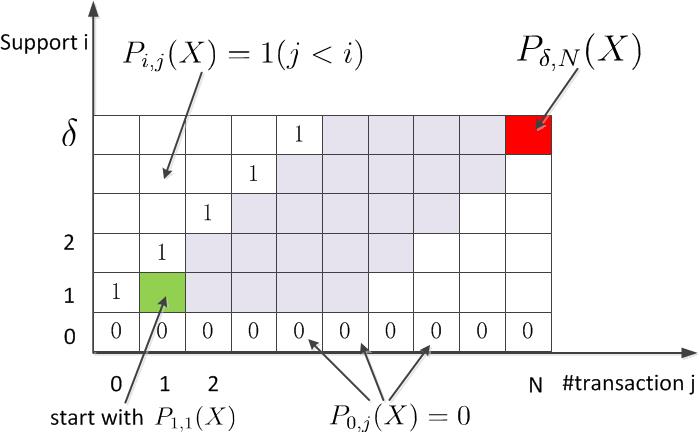}\\
  \caption{Dynamic computation process for  $P_{< \delta, N}(X)$}\label{fig:probdynamic}
\end{figure}

In Figure~\ref{fig:probdynamic}, the computation process is started from the green grid as $P_{1,1}(X)$ and ended with the right top red grid $P_{\delta, N}(X)$. According to the analysis in \cite{bernecker2009probabilistic}, the computation process  requires at most $O(\delta \cdot N)$ time complexity and $O(N)$ space complexity.

\textbf{\emph{General.}} $X$ as a multiple item composed itemset, $|X| >1$. In this case, the object $Pr[\bigwedge_{z\in X}  0  <sup(z) <  \delta ]$ can be notated as:
$$P_{\underbrace{<i,\dots,<i}_{|X|}; j}({z_1,\dots,z_{|X|}})$$
where $z_1,\dots,z_{|X|}$ is single item and $z_1,\dots,z_{|X|} \in X$.
Thus the recursive relationship is modified as follows:
\begin{align}
\begin{split}
&P_{\underbrace{<i,\dots,<i}_{|X|}; j}({z_1,\dots,z_{|X|}})\\
&=P_{\underbrace{<i-1,\dots,<i}_{|X|}; j}({z_1,\dots,z_{|X|}}) \cdot P(z_1 \subseteq t_j)\\
 & + P_{\underbrace{<i,\dots,<i}_{|X|}; j}({z_1,\dots,z_{|X|}}) \cdot P(z_1 \nsubseteq t_j)\\
 & \qquad \cdots \\
 & +P_{\underbrace{<i,\dots,<i-1}_{|X|}; j}({z_1,\dots,z_{|X|}}) \cdot P(z_{|X|} \subseteq t_j)\\
 & + P_{\underbrace{<i,\dots,<i}_{|X|}; j}({z_1,\dots,z_{|X|}}) \cdot P(z_{|X|} \nsubseteq t_j)\\
    \end{split}
\end{align}

where the boundary cases turn into:
\begin{equation*}
\begin{cases}
P_{<i_1,\dots,<i_{|X|}; j}({z_1,\dots,z_{|X|}})=1, \qquad j< min\{i_1, \dots, i_{|X|}\}\\
P_{\underbrace{<0,\dots,<0}_{|X|}; j}({z_1,\dots,z_{|X|}})=0, \qquad 0 \leq j \leq N\\
\end{cases}
\end{equation*}

We can build dynamic programming based computation process similar to Figure~\ref{fig:probdynamic} as general high-dimensional case of the one-dimensional case in Figure~\ref{fig:probdynamic}. We note that such general  computation process  requires at most $O(\delta^{|X|} \cdot N)$ time complexity and $O(N^{|X|})$ space complexity. Here we neglect the details due to the space limitation.

We note that such basic  verification mechanism  illustrated above can be used to verify the random faults caused by a honest cloud, and the stupid cloud attack. However, once a smart cloud aware of this mechanism as before, and returns the elaborated fake values as the results, the basic verification scheme will also fail. Therefore, we doubt why we can build some enhanced verification mechanism to resolve this problem and resist such smart cloud attack. Below we will explore and discuss it in details.

\subsection{Enhanced Checkers with Private Weights}
Now we try to explore the potential similar private weights based  enhancement  for PWS case below. After we assign the weights and try to associate the possibilities of the itemset between original database $T$ and virtual transformed $T'$, unluckily, we find it difficult due to the properties of the PWS. That is to say, it is hard to find some simple mathematic relationship between $X.pcnt_{T'}$ and $X.pcnt_{T}$:
\begin{equation}
\label{equ:virtualtrans}
  X.pcnt_{T'}=\phi(w_z)_{z \in X}X.pcnt_{T}
\end{equation}

where function $\phi$ is the object function with the weights $w_z, z \in X$ as its operands. The private weights enhancement schemes holds in Theorem~\ref{the:privatedeterverify}, \ref{the:expfirstverify}, and \ref{the:expsecondverify} with some specific qualified function $\phi$ due to the simple property of the sum of support and expected support. And the good representation property in Equation~\eqref{equ:virtualtrans} is the crux of the success of the virtual transformation, in other words, we don't need actual transformation from $T$ to $T'$,  instead  we can  avoid the transformation and compute the values $X.pcnt_{T'}$ from $X.pcnt_{T}$ according to  Equation~\eqref{equ:virtualtrans}.

We illustrated the problem with $X=A$, where $A$ is a single item, then in the basic verification scheme, if the item $A$ is UFI, then its frequent value will be computed and returned as the form below:
$$ A.pcnt_{T}=Pr[sup(A) \geq \delta] =\sum_{L = \delta}^N Pr[sup(A)=L]=\sum_{L = \delta}^N p_L$$

where $p_L$ denotes the sum of the possible worlds' value, in which the support of the item $A$ is $L$. We should mention that the cloud will return the sum $A.pcnt_{T}$ instead of the $p_L$s.

Now if we apply the similar weighted enhancement verification, the item $A$ assigned by the weight $w_A$, then
$$ A.pcnt_{T'} =\sum_{L = \delta}^N w_A^LPr[sup(A)=L]=\sum_{L = \delta}^N w_A^Lp_L$$

Here we compare the forms of the two equations above, the relationship with the weight $w_A$ is hard to obtain:
\begin{equation*}
  A.pcnt_{T'}=\phi(w_A)A.pcnt_{T}
\end{equation*}

Thus in order to finish the enhanced verification scheme, we have to compute the $p_L$ value then perform the transformation $w_A^Lp_L$, even we can get the closed form of $\sum_{L = \delta}^N p_L$ or $\sum_{L = \delta}^N w_A^Lp_L$ if luck enough. 
But in fact, unfortunately,  it is hard to obtain the  closed form of $\sum_{L = \delta}^N p_L$ or $\sum_{L = \delta}^N w_A^Lp_L$ with better complexity compared with the original  exponential complexity problem if using brute-force enumeration.

Therefore, we can get the conclusion in this section as \textbf{\emph{Remarks: in PWS based outsourcing UFI mining scenario, there exists effective basic checker mechanism, and its enhanced private weights based resolutions can't be built efficiently enough.
Though it seems a little upsetting, such basic verification mechanism can be used to verify the random faults caused by a honest cloud, which is still useful for us. }}
Besides, we can still find some more good news from the case of approximate UFI mining method in the next section, which will act as the bridge of the expected support based UFI definition and PWS based UFI definition.

\section{ Verification for Approximate UFI mining}
Due to complexity of the exact probabilistic frequent algorithms, when uncertain databases are large enough,   the itemsets' support follow Poisson Binomial distribution, so that the approximate algorithms can obtain the approximate frequent probability with high quality by only acquiring the basic statistic itemset information, including support expectation, variance, with $ O(N)$ computation cost and reach satisfactory result.  Specifically, in \cite{wang2010accelerating}, the authors proposed the Poisson distribution-based approximate probabilistic frequent itemset mining algorithm, called PDUApriori.
However, this algorithm only approximately determines whether an itemset is probabilistic frequent, and cannot return accurate frequent probability values.
The Normal distribution-based approximate probabilistic frequent itemset mining algorithm, NDUApriori, was proposed in \cite{calders2010approximation}. 
However, it is impractical to large sparse uncertain databases since it employs the Apriori framework according to \cite{tong2012mining}.
In order to solve this problem, \cite{tong2012mining} proposed  NDUH-Mine to integrates the framework of UH-Mine and the Normal distribution approximation in order to achieve a win-win partnership in sparse uncertain databases.

It is claimed in \cite{tong2012mining} that  the Normal distribution-based approximation algorithms build a bridge between the two  different UFI definitions, expected support-based frequent itemsets and the probabilistic frequent itemsets definition. In detail, according to \cite{calders2010approximation},
$$P(sup(X) \geq \delta )\approx \Phi(\frac{\delta-0.5-esup(X)}{\sqrt{Var(X)}})$$
where $\Phi$ is the cumulative distribution function of standard Normal distribution, and $Var(X)$ is the variance of the support of $X$. Though there is no outsourcing  approximate UFI mining scheme, we can see that the crux of the outsourcing approximate UFI mining scheme is the computation of the expected support $esup(X)$ and the variance  $Var(X)$. Thus the verification method should complete the verification of right computation of $esup(X)$ and $Var(X)$ in the remote cloud. Below we discuss the details.

Obviously, the verification of $esup(X)$ has been completed in Section~\ref{sec:exp} in expected support UFI definition case. The private weights based enhancement can be utilized well. Below we study the verification of  $Var(X)$. Due to the relationship between $esup(X)$ and $Var(X)$,
\begin{equation}
\label{equ:variance}
  Var(X)=esup(X^2)-esup(X)^2
\end{equation}

Here we should mention that part $esup(X^2)$ represents the expected support of the itemset when the existence probability squares from the original values, e.g., the original $P(a \in t_i)=a_i$ changes into $ a_i^2$. Thus once we get the returned values $esup(X)$ and $Var(X)$. We first finish the verification of $esup(X)$, then  verify the $esup(X^2)$ on the transformed database $T': T \to T^2$ using the same methods in Section~\ref{sec:exp}.

It is worth mentioning that such transformation $T': T \to T^2$ is necessary because we can't establish the required function $\phi$ in Equation~\eqref{equ:virtualtrans} even in expected support UFI definition 
, which is the crux of the free actual transformation. The private weights based enhancement also applies here. Due to the verification of $esup(X)$ and $Var(X)$ is combined with the $P(sup(X) \geq \delta )= X.pcnt$, thus the approximate computation result can be ensured by some accuracy.

The process can be illustrated as in Algorithm~\ref{alg:approximation}.
\begin{algorithm}
\SetKwInOut{Input}{Input}\SetKwInOut{Output}{Output}
\SetCommentSty{it}
\begin{enumerate}
  \item[(1)] verify $esup(X)$ on original database $T$ according to the methods in Section~\ref{sec:exp}\;
  \item[(2)] transform database $T': T \to T^2$ and prepare verification values $esup(X^2)=Var(X)-esup(X)^2$\;
  \item[(3)]  verify $esup(X^2)$ on transformed database $T'$ according to the methods in Section~\ref{sec:exp}\;
\end{enumerate}
\caption{ verification steps for outsourcing approximate UFI mining  with returned $Var(X)$, $esup(X)$}
\label{alg:approximation}
\end{algorithm}

\section{Analysis}
We present the comparisons and analysis on the schemes proposed in this paper  in Table~\ref{tab:analysis}. All the cost values are compared with the deterministic basic scheme in 1st line.  We note that in expected support-based outsourcing FI mining scenario, there exists effective basic checker mechanism, and its enhanced private weights based resolutions can also be built effectively, which can detect all the cloud attacks including  Random Fault, Smart/Stupid attack. In  PWS based outsourcing UFI mining scenario, there exists effective basic checker mechanism, however, its enhanced private weights based resolutions can't be built efficiently enough. Though it seems a little upsetting, such basic verification mechanism can be used to verify the random faults caused by a honest cloud, which is still useful for us. After that, we can note the further conclusion that  the scenario of existing approximation UFP mining, where we can see that our technique can provide good probabilistic guarantees about the correctness of the verification with some extra cost effectively. 
\begin{table*}
\center
  \begin{tabular}{|c|c|c|c|c|}
  \hline
    Scenario & Level & Extra cost & Random Fault/ & Smart  \\
    & & & Stupid cloud attack & cloud attack \\
    \hline
  Deterministic  & basic & - & Y  & N \\
   & weight &  - & Y & Y \\
    \hline
   Expected UFI  & basic & - & Y & N \\
  & weight1 &  - & Y & Y \\
    & weight2 &  - & Y & Y \\
  \hline
   PWS UFI  & basic &  time :$O(\delta^{|X|} \cdot N)$& Y & N \\
   & & space:  $O(N^{|X|})$ & & \\
     \hline
  Approximate  & b/w & time :$O(N\cdot |I|)$ & Y & Y \\
   \hline
\end{tabular}
 \label{tab:analysis}
\end{table*}

\section{Clonclusion}
In this paper, we extend the existing outsourcing FI mining work to uncertain area w.r.t. the two popular UFI definition criteria and the approximate UFI mining methods.  Specifically, we design the basic checker mechanism verification for expected support UFI definition to verify the random fault, then propose two enhanced schemes based on the private random weights mechanism for the sake of smart cloud attack verification.
Then we explore and design the basic/enhanced verification scheme for PWS based UFI definition, which is able to verify random fault/smart attack. We address the efficiency and  difficulty of the enhanced scheme for PWS case.
As the bridge of the two different UFI  definition,  we further discuss the case of  approximation UFI mining  verification, where we  reduce the verification to the expected support verification and  provide good probabilistic guarantees about the correctness of  verification. Finally, we present the comparisons and analysis on the schemes proposed in this paper.

\section*{Acknowledgment}
This work was supported in part by the National Natural Science Foundation of China (No.61170233, No.61232018, No.61272472, No.61202404, No.61272317)
and China Postdoctoral Science Foundation (No.2011M501060).


\bibliographystyle{plain}
\bibliography{arXiv_main}

\end{document}